\documentclass[journal]{IEEEtran}
\usepackage[font={small}]{caption}
\usepackage{epsfig,amsmath,amssymb,epsf,amsthm,scalefnt,multirow,subfig,dsfont}
\usepackage{xcolor}
\usepackage{float}
\usepackage{cite}
\usepackage{psfrag}
\usepackage{mathtools}
\usepackage{lipsum}
\usepackage{cuted}
\usepackage{multicol}
\usepackage{amsmath,centernot}

\newcommand{\nll}{\centernot{\ll}}

\newtheorem{theorem}{Theorem}
\newtheorem{remark}{Remark}
\newtheorem{lemma}{Lemma}
\newtheorem*{lemma*}{Lemma}

\newtheorem{definition}{Definition}

\def \supp{\operatorname{supp}}

\def\b0{{\pmb{0}}} 

\allowdisplaybreaks[4]
\begin{document}

\title{Treating Interference as Noise is Optimal for Covert Communication over Interference Channels}

 \author{\IEEEauthorblockN{Kang-Hee Cho and Si-Hyeon Lee}\\
\IEEEauthorblockA{School of Electrical Engineering, KAIST,  South Korea \\
E-mail:kanghee@kaist.ac.kr, sihyeon@kaist.ac.kr
}\thanks{This paper was submitted to IEEE Transactions on Information Forensics and Security, and a shorter version of this paper was submitted to IEEE ISIT 2020 \cite{cholee_ISIT20}.}}

\maketitle

\begin{abstract}
We study the covert communication over $K$-user discrete memoryless interference channels (DM-ICs) with a warden. It is assumed that the warden's channel output distribution induced by $K$ ``off" input symbols, which are sent when no communication occurs, is not a convex combination of those induced by any other combination of input symbols (otherwise, the square-root law does not hold). We derive the exact covert capacity region and show that a simple point-to-point based scheme with treating interference as noise is optimal. In addition, we analyze the secret key length required for the reliable and covert communication with the desired rates, and present a channel condition where a secret key between each user pair is unnecessary. The results are extended to the Gaussian case and the case with multiple wardens. 

\end{abstract}

\section{Introduction}\label{sec:intro}
The  covert communication or communication with \textit{low probability of detection} aims to ensure a reliable communication between legitimate parties while keeping the presence of the communication secret from the warden. This setup is applicable to the military situation where several units of allies want to communicate each other without being detected by the enemy (i.e., the warden). The fundamental limits of covert communications have been actively studied mainly for point-to-point (p-to-p) channels such as additive white Gaussian noise (AWGN) channels \cite{Bash:13, Wang:16}, discrete memoryless channels (DMCs) \cite{Bloch:16, Wang:16},  low-complexity coding scheme based on the pulse-position modulation \cite{PPM1, PPM2},  channels using multiple antennas \cite{MIMOAWGN}, and channels with some uncertainty of statistics \cite{csit, noncoherentjournal}. In most interesting cases, the covertness constraint restricts the number of no ``off" input symbols (for discrete channel cases \cite{Bloch:16}) or the transmit power (for continuous channel cases \cite{Bash:13, Wang:16}) that leads the so-called square-root law, i.e., the maximum number of bits that can be communicated reliably and covertly over $n$ channel uses scales proportionally to $\sqrt{n}$.  

Recently, this line of research has been extended to various network scenarios such as multiple access channels (MACs) \cite{MAC}, broadcast channels (BCs) \cite{BC}, relay channels (RCs) \cite{covertrelay, helprelay, fullrelay}, and wireless adhoc networks \cite{adhoc}. It turns out that in some canonical models, the covertness constraint affects optimal strategies and/or the form of capacity region. For the DM-MAC with a warden \cite{MAC}, it is shown that the capacity region has no sum-rate bound and time-sharing is not needed to achieve the capacity region, both in contrast to the case without a warden \cite{Liao}. For the DM-BC with a warden \cite{BC},  a simple time-division approach is shown to be optimal over some channels satisfying a certain condition that contains a broad class of channels where the capacity region is not known without a warden \cite{CoverBC, Gallager:74}. 
 
In this paper, we consider another important network scenario, the $K$-user discrete memoryless interference channel (DM-IC) with a warden. The warden monitors its channel outputs through a DM-MAC. We assume that there is an ``off" input symbol at each transmitter (Tx) that is sent when no communication occurs. Then, we focus on the case that the warden's output distribution induced by $K$ ``off" input symbols is not a convex combination of some other output distributions at the warden; otherwise, the square-root law does not hold. In the absence of the covertness constraint, the capacity region of DM-ICs is not known in general except some special cases e.g., strong ICs \cite{strong} and injective deterministic ICs \cite{deterministic}. In addition, to obtain the best known inner bound (Han-Kobayashi inner bound \cite{Han-Kobayashi}), somewhat complicated coding strategies such as rate-splitting and superposition coding are utilized. In the presence of the covertness constraint, we derive the exact covert capacity region of the $K$-user DM-IC. Interestingly, an optimal strategy is  shown to be p-to-p-based scheme with treating interference as noise (TIN). We also analyze the secret key length required for the reliable and covert communication with the desired rates by using channel resolvability approach \cite{Han-Verdu, spectrum, Bloch:16, MAC}, and  derive the channel condition where a secret key is not required to be shared between each user pair. 

For brevity of the presentation, we first consider the binary input (BI) DM-IC in Section \ref{sec:problem} to Section \ref{sec:proof}. For the BI DM-IC with a warden, we fomulate the problem in Section \ref{sec:problem} and present the covert capacity region  in Section \ref{sec:results}, which is proved in Section \ref{sec:proof}. The results are extended to the non-binary input case, to the Gaussian channels, and to the channels with $J$ wardens  in Section \ref{sec:exten}. Finally, we conclude our paper in Section \ref{sec:concl}.

\textbf{Notation}: The notation for this paper is summarized as follows. To represent random variables and their realizations, we use upper case (e.g., $X$) and lower case (e.g., $x$), respectively. For length $n$ random vectors related to a random variable, we use boldface (e.g., $\mathbf{X}$ and $\mathbf{x}$). We define the set $\mathcal{K} \coloneqq [1:K]:=\{1,\cdots, K\}$ for a positive integer $K \geq 2$. For $\mathcal{U} \subseteq \mathcal{K}$, we denote the vector $\{X_k: k \in \mathcal{U}\}$ as $X_{\mathcal{U}}$, and the cartesian product $\times _{k \in \mathcal{U}} \mathcal{X}_k$ as $\mathcal{X}_{\mathcal{U}}$. We denote entropy function of a random variable $X$ as $H(X)$ and differential entropy function as $h(X)$. Relative entropy and variational distance are denoted as $D(P \| Q) \coloneqq \sum_x P(x) \log \frac{P(x)}{Q(x)}$ and $\mathbb{V}(P,Q) \coloneqq \frac{1}{2}\sum_x | P(x) - Q(x) |$, respectively. Mutual information of $(X,Y) \sim P \times W$ is denoted by $I(X;Y)$ and $I(P,W)$. We define $[x]^+ \coloneqq \max (x, 0)$, and $\mathds{1}\{\cdot\}$ is the indicator function. 
We denote the support of the probability distribution $P$ by $\supp(P)$.
For two probability mass functions $P$ and $Q$ that are defined on the same alphabet $\mathcal{Z}$, we write $P \ll Q$ if $P$ is absolutely continuous with respect to $Q$, i.e., $Q(z) = 0$ implies $P(z) = 0$ for all $z \in \mathcal{Z}$.

\section{Problem Formulation}\label{sec:problem}
\begin{figure}
 \centering
  {
  \includegraphics[width=88mm]{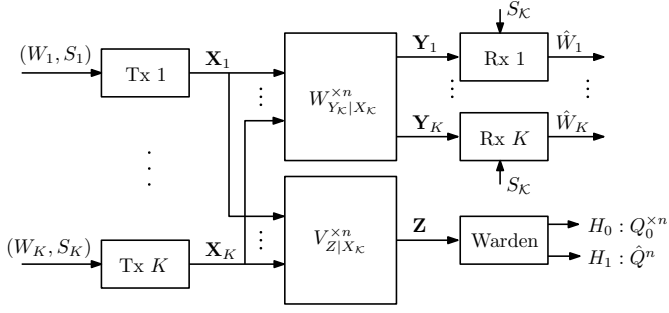}}
  \caption{A $K$-user DM-IC with a warden} \label{fig:channel}
\end{figure}
Consider a covert communication scenario over a $K$-user DM-IC with a warden depicted in Fig. \ref{fig:channel}. Through a DM-IC $(\mathcal{X}_{\mathcal{K}}, W_{Y_{\mathcal{K}}|X_{\mathcal{K}}}, \mathcal{Y}_{\mathcal{K}})$ that consists of $K$ channel input alphabets $\mathcal{X}_{\mathcal{K}}$, a channel transition matrix $W_{Y_{\mathcal{K}}|X_{\mathcal{K}}}$, and $K$ channel output alphabets $\mathcal{Y}_{\mathcal{K}}$, each user pair $k$ wants to communicate the message $W_{k}$ reliably, while keeping the presence of the communication secret from the warden who observes its channel outputs through a DM-MAC $(\mathcal{X}_{\mathcal{K}}, V_{Z|X_{\mathcal{K}}}, \mathcal{Z})$ where $V_{Z|X_{\mathcal{K}}}$ is the channel transition matrix and $\mathcal{Z}$ is the channel output alphabet at the warden. 
For brevity, we let $\mathcal{X}_k = \mathcal{X} = \{ 0,1 \}$ for all $k \in \mathcal{K}$. 
We also let $0 \in \mathcal{X}$ be the ``off'' input symbol that is sent when no communication occurs.
The marginal channel at receiver (Rx) $k$ is denoted as $W_{Y_k|X_{\mathcal{K}}}$.
In addition, we denote a channel submatrix $W_{Y_k|X_{\mathcal{U}} =  x_{\mathcal{U}}, X_{\mathcal{K}\backslash \mathcal{U}}}$ (i.e., the marginal channel $W_{Y_k|X_{\mathcal{K}}}$ when $X_{\mathcal{U}}$ is fixed to $x_{\mathcal{U}}$) as $W_{Y_k|x_{\mathcal{U}}, X_{\mathcal{K}\backslash \mathcal{U}}}$ for brevity.
We denote by  $b(\mathcal{U}) \in \mathcal{X}_{\mathcal{K}}$ for $\mathcal{U} \subseteq \mathcal{K}$  the length-$K$ binary vector where the $k^{\mathrm{th}}$ component is $\mathds{1}\{k \in \mathcal{U} \}$.
For notational convenience, we define $W_{\mathcal{U}}^{(k)}(y) \coloneqq W_{Y_k|X_{\mathcal{K}}}(y|b(\mathcal{U}))$ and $Q_{\mathcal{U}}(z) \coloneqq V_{Z|X_{\mathcal{K}}}(z|b(\mathcal{U}))$. 
If $\mathcal{U} = \emptyset$, i.e., no communication takes place, we write $W_{0}^{(k)}$ and $Q_0$, and if  $\mathcal{U} = \{i\}$ for $i\in \mathcal{K}$, we write $W_{i}^{(k)}$ and $Q_i$ for brevity.

In the following, we define a sequence of codes for our covert communication setting.
\begin{definition}
An $(M_{\mathcal{K}}, J_{\mathcal{K}}, n)$ code for the $K$-user DM-IC with a warden consists of
\begin{itemize}
\item $K$ message sets $[1 : M_k]$ for $k \in \mathcal{K}$;
\item $K$ secret key sets $[1 : J_k]$ for $k \in \mathcal{K}$;
\item $K$ Txs $\mathbf{x}_k: [1 : M_k] \rightarrow \mathcal{X}^n$ for $k \in \mathcal{K}$, where each Tx $k$ encodes message-key pair $(w_k, s_k) \in [1 : M_k] \times [1 : J_k]$ as a length-$n$ codeword $\mathbf{x}_k(w_k, s_k)$;
\item $K$ Rxs $\hat{w}_k : \mathcal{Y}_k^n \times (\times_{k \in \mathcal{K}}[1 : J_k]) \rightarrow [1 : M_k]$ for $k \in \mathcal{K}$, where each Rx $k$ estimates the message as $\hat{w}_k$ based on its channel outputs $\mathbf{y}_k$ and the secret keys $s_\mathcal{K}$.    
\end{itemize}
\end{definition}
Each message-key pair $(W_k, S_k)$ is uniformly distributed over $[1 : M_k] \times [1 : J_k]$. The probability of decoding error is defined as $P_e^n \coloneqq \Pr \left( \bigcup_{k=1}^K \{ \hat{W}_k \neq W_k \} \right)$.

When the communication takes place, the warden observes its channel outputs $\mathbf{Z} \in \mathcal{Z}^n$ of which distribution is given as 
\begin{align}
\hat{Q}^n(\mathbf{z}) \coloneqq \frac{1}{\prod_{k \in \mathcal{K}} M_k J_k}\sum_{w_{\mathcal{K}}}\sum_{s_{\mathcal{K}}} V_{Z|X_{\mathcal{K}}}^{\times n} (\mathbf{z} | \mathbf{x}_{\mathcal{K}}(w_{\mathcal{K}},s_{\mathcal{K}})).
\end{align}
When no communication occurs, $\mathbf{Z}$ is distributed according to $Q_{0}^{\times n}$, the $n$-fold product distribution of $Q_{0}$. 
Based on the channel statistic and the channel outputs, the warden performs a hypothesis test to determine whether the communication takes place (hypothesis $H_1$) or not (hypothesis $H_0$). 
The covert communication requires that the sum of the probabilities of false alarm (accept $H_1$ when no communication occurs) $\pi_{1|0}$ and miss detection (accept $H_0$ when the communication occurs) $\pi_{0|1}$ is close to $1$ (corresponding to a blind test). The optimal hypothesis test of the warden satisfies
\begin{align} \label{eqn:opthyp}
\pi_{1|0} + \pi_{0|1} &= 1 - \mathbb{V}(\hat{Q}^n, Q_{0}^{\times n}) \\ \label{eqn:Pinsker}
&\geq 1 - \sqrt{D ( \hat{Q}^n \| Q_{0}^{\times n} )},
\end{align} 
where \eqref{eqn:opthyp} can be checked in \cite{hypo}, and \eqref{eqn:Pinsker} follows by Pinsker's inequality \cite{Cover:2006}. 
Thus, we set the covertness constraint as follows:
\begin{align} \label{eqn:covertness}
\lim_{n \rightarrow \infty} D ( \hat{Q}^n \| Q_{0}^{\times n} ) = 0.
\end{align}
We assume $Q_{\mathcal{U}} \ll Q_0$ for all $\mathcal{U} \subseteq \mathcal{K}$. It can be easily seen that if $Q_{\mathcal{U}} \nll Q_0$ for some $\mathcal{U} \subseteq \mathcal{K}$, the relative entropy in \eqref{eqn:covertness} is infinity when the Txs send symbols $b(\mathcal{U})$. 
We also assume $W_{\mathcal{U}}^{(k)} \ll W_{0}^{(k)}$ for all $k \in \mathcal{K}$ and for all $\mathcal{U} \subseteq \mathcal{K}$\footnote{Some comments on the scenario without this assumption are in Remark~\ref{rem:absolute}.}. Furthermore, we assume that $Q_0$ cannot be represented as any convex combination of $Q_{\mathcal{U}}$ for some $\mathcal{U} \subseteq \mathcal{K}$; otherwise, one can design a sequence of codes that the relative entropy in \eqref{eqn:covertness} is zero while achieving a positive rate (i.e., the square-root law does not hold).

The covert capacity region is formally defined in the following.
  \begin{definition}
We say that a tuple pair $ (R_{\mathcal{K}}, L_{\mathcal{K}}) \in \mathbb{R}_{+}^{2K}$ is achievable for the $K$-user DM-IC with a warden if there exists a sequence of codes satisfying the following:
\begin{align}
\liminf_{n \rightarrow \infty} \frac{\log M_k}{\sqrt{nD ( \hat{Q}^n \| Q_{0}^{\times n} )}} &\geq R_k, \quad \forall k \in \mathcal{K},
\end{align}
\begin{align}
\limsup_{n \rightarrow \infty} \frac{\log J_k}{\sqrt{nD ( \hat{Q}^n \| Q_{0}^{\times n} )}} \leq L_k, \quad \forall k \in \mathcal{K},
\end{align}
\begin{align} \label{eqn:error0}
\lim_{n \rightarrow \infty} P_e^n = 0, 
\end{align}
and
\begin{align} \label{eqn:covert0}
\lim_{n \rightarrow \infty} D ( \hat{Q}^n \| Q_{0}^{\times n} ) = 0.
\end{align}
The covert capacity region of the $K$-user DM-IC with a warden is defined as the closure of the set $ \{ R_{\mathcal{K}} \in \mathbb{R}_{+}^{K} : (R_{\mathcal{K}},L_{\mathcal{K}}) \mbox{ is achievable for some } L_{\mathcal{K}} \}$.
\end{definition}

\section{Main Results}\label{sec:results}
In this section, we present our main theorem on the covert capacity region of the $K$-user DM-IC with a warden. Furthermore, we provide a sufficient and necessary condition on the secret key length at the boundary of the covert capacity region from which we can obtain the channel condition where a secret key is not required to be shared. The proof of the main theorem is in Section \ref{sec:proof}.

\begin{theorem}\label{thm:region}
For the $K$-user DM-IC with a warden, the covert capacity region is the set of the rate tuple $R_{\mathcal{K}}$ satisfying 
\begin{align} \label{eqn:regionthm}
R_k \leq \frac{\alpha_k D(W_{k}^{(k)} \| W_{0}^{(k)})}{\sqrt{\chi^2(\boldsymbol{\alpha})/2}}, \quad \forall k \in \mathcal{K}
\end{align}
for some $\boldsymbol{\alpha} \in [0,1]^K$ such that $\sum_{k \in \mathcal{K}}\alpha_k = 1$, where $\chi^2(\boldsymbol{\alpha})$ is defined as
\begin{align}
\chi^2(\boldsymbol{\alpha}) \coloneqq \sum_z \frac{\left( \sum_{k \in \mathcal{K}} \alpha_k Q_k(z) - Q_0(z) \right)^2}{Q_0(z)}.
\end{align}
For $R_{\mathcal{K}}$ satisfying  \eqref{eqn:regionthm} with equalities, a sufficient and necessary condition on the tuple $L_{\mathcal{K}}$ for $(R_{\mathcal{K}}, L_{\mathcal{K}})$ to be achievable is 
\begin{align} \label{eqn:ratekeyrate}
L_k \geq \frac{\alpha_k [D(Q_{k} \| Q_{0}) - D(W_{k}^{(k)} \| W_{0}^{(k)})]^{+}}{\sqrt{\chi^2(\boldsymbol{\alpha})/2}}, \quad \forall k \in \mathcal{K}. 
\end{align}
Thus, if $D(Q_k\|Q_0) \leq D(W_{k}^{(k)} \| W_{0}^{(k)})$ (i.e., roughly the channel from Tx $k$ to the warden is worse than the channel from Tx $k$ to Rx $k$), a secret key between user pair $k$ is unnecessary.  
\end{theorem}

For the achievability, the codebook is randomly generated with a very low probability of sending symbol $1$ (approximately order of $1/\sqrt{n}$), and each Rx $k$ decodes the message by treating interference as noise. Thus, each secret key $S_k$ is required to be shared between only each user pair $k$.
In the following, a few remarks on Theorem \ref{thm:region} are in order. 
\begin{enumerate}
\item It turns out that the covertness requirement constrains the number of symbol $1$ at each Tx in a certain way. Roughly speaking, the vector $\boldsymbol{\alpha}$ represents how we allocate the number of symbol $1$ to each Tx, but the total number of symbol $1$ sent from all the Txs depends on the common factor $\chi^2(\boldsymbol{\alpha})$. The factor $\chi^2(\boldsymbol{\alpha})$ appears because the transmission from each Tx \emph{jointly} affects the covertness constraint. The factor  $\chi^2(\boldsymbol{\alpha})$ varies with $\boldsymbol{\alpha}$ in general because the channels from the each Tx to the warden are different (i.e., each Tx's symbol $1$ influences the dectectability of the warden differently).
 
\item The optimality of TIN can be explained as follows. For general ICs without a warden, every user cannot simultaneously achieve the maximally achievable individual rate because one user's transmission interferes the other users channel. However, for our model with a warden, since the influence on the warden's channel outputs is kept negligible by restricting the number of symbol $1$, the effect of the interference signals on each Rx is also negligible. Thus, every user can achieve the maximally achievable individual rate, given that a certain fraction $\boldsymbol{\alpha}$ of symbol $1$ is allocated to them. 

\item We remind that the factor  $\chi^2(\boldsymbol{\alpha})$ varies with $\boldsymbol{\alpha}$ in general. However, if the DM-MAC $(\mathcal{X}_{\mathcal{K}}, V_{Z|X_{\mathcal{K}}}, \mathcal{Z})$ is symmetric in the sense that $Q_k(z) = Q(z)$, $\forall k \in \mathcal{K}$ and $\forall z \in \mathcal{Z}$, $\chi^2(\boldsymbol{\alpha})$ is invariant in $\boldsymbol{\alpha}$. Hence, the time-division approach is optimal under this condition.   

\end{enumerate}

\begin{remark}\label{rem:absolute}
For p-to-p DMCs, where the output distributions at the Rx induced by symbols $0$ and $1$ are denoted by $P_0$ and $P_1$, respectively, the optimal covert communication for the case of  $P_1 \nll P_0$ is well-studied \cite[Appendix G-F]{Bloch:16}. Let $p$ denote the probability of sending symbol $1$. If $P_1 \nll P_0$, by utilizing the positions of the channel outputs that belong to $\supp(P_1) \backslash \supp(P_0)$,  the order of the optimal throughput over $n$ channel uses increases to approximately $\sqrt{n}\log n$. In this case, the optimal  $p$ turns out to be order of $1/n$. We remind that if $P_1 \ll P_0$, the square root law holds and the optimal $p$ is approximately order of $1/\sqrt{n}$. 

It is not straightforward to generalize the aforementioned result to our model. Consider $K=2$ and assume that $W_{1}^{(1)} \nll W_{0}^{(1)}$ and $\supp(W_{1}^{(1)}) = \supp(W_{2}^{(1)})$. Then, for Rx $1$, it is not clear to infer which Tx sends symbol $1$ by just observing a certain channel output, and thus some joint decoding scheme might be needed to utilize the advantage of $W_{1}^{(1)} \nll W_{0}^{(1)}$. Furthermore, it is not straightforward what the order of optimal input distributions should be, as a Tx input influences not only the channel output of its corresponding Rx, but also that of the other Rx. 

\end{remark} 

\section{Proof of Theorem \ref{thm:region}}\label{sec:proof}

Let us first define some probability distributions that will be used throughout this paper. We define $K$ Bernoulli distributions $\{ P_k\}_{k \in K}$ for $\gamma_n \in [0,1]$ and $\boldsymbol{\alpha} \in [0,1]^K$ such that $\sum_{k \in \mathcal{K}}\alpha_k = 1$  as follows:
\begin{align} 
P_{k}(x) \coloneqq 
\begin{cases}
1- \alpha_k \gamma_n & x = 0 \\
 \alpha_k \gamma_n & x = 1.
\end{cases} \label{eqPk}
\end{align}
Then, we define the channel output distributions  induced by the input distribution $P_k$ at each Tx $k$ as the following:
\begin{align} 
W_{\boldsymbol{\alpha},\gamma_n}^{(k)}(y_k) \coloneqq \sum_{x_{\mathcal{K}}} W_{Y_k | X_{\mathcal{K}}}(y_k | x_{\mathcal{K}}) \left( \prod_{k \in \mathcal{K}}P_k(x_k) \right), \\ \label{eqn:wardendist}
Q_{\boldsymbol{\alpha},\gamma_n}(z) \coloneqq \sum_{x_{\mathcal{K}}} V_{Z | X_{\mathcal{K}}}(z | x_{\mathcal{K}}) \left( \prod_{k \in \mathcal{K}}P_k(x_k) \right).
\end{align} 

\subsection{Relation to the DM-MAC with a Warden \cite{MAC}}
In our scenario, the warden observes its channel outputs through a DM-MAC. This channel structure from the Txs to the warden is same to that of the DM-MAC with a warden \cite{MAC}, and thus some results in  \cite{MAC} on the influence of the transmissions on the warden's induced channel outputs apply to our scenario. For brevity of the description, we omit the details of the proofs that are same to that of \cite{MAC}, but we provide some comments to help understanding.       

We introduce the following lemma \cite[Lemma 1]{MAC}, which is proved for the DM-MAC with a warden and also holds for our setting. This lemma presents an important result on how the number of symbol $1$ at each Tx affects the relative entropy and mutual information of interst, which is used to prove the achievability part.
\begin{lemma}[Arumugam-Bloch \cite{MAC}] \label{lem:basic}
Let $\{\gamma_n\}_{n \geq 1}$ be a sequence such that $\gamma_n \in [0,1]$ and $\lim_{n \rightarrow \infty}\gamma_n = 0$. Then, for sufficiently large $n$ and $\boldsymbol{\alpha} \in [0,1]^K$ such that $\sum_{k \in \mathcal{K}}\alpha_k = 1$, 
\begin{align} \label{eqn:keylemma}
\frac{\gamma_n^2}{2}(1 - \sqrt{\gamma_n})\chi^2(\boldsymbol{\alpha}) \leq D(Q_{\boldsymbol{\alpha},\gamma_n} \| Q_0) \nonumber \\\leq \frac{\gamma_n^2}{2}(1 + \sqrt{\gamma_n})\chi^2(\boldsymbol{\alpha}).
\end{align}
Futhermore, for any tuple $(X_{\mathcal{U}}, Z)$ for $\mathcal{U} \subseteq \mathcal{K}$ and $|\mathcal{U}| \neq 0$, with the joint distribution $V_{Z|X_{\mathcal{U}}}(\prod_{k \in \mathcal{U}}P_k)$, we have
\begin{align}
I(X_{\mathcal{U}} ; Z) = \sum_{k \in \mathcal{U}} \alpha_k \gamma_n D(Q_k \| Q_0) + O(\gamma_n^2).
\end{align}  
\end{lemma}

\subsection{Achievability}\label{subsec:achieve}
For the achievability, we use the TIN scheme based on the standard random coding argument. The achievable secret key length is analyzed based on the channel resolvability approach \cite{Han-Verdu, spectrum, Bloch:16, MAC}.  

\subsubsection{Treating Interference as Noise}\label{subsubsec:scheme}
Fix  $\boldsymbol{\alpha} \in [0,1]^K$ such that $\sum_{k \in \mathcal{K}}\alpha_k = 1$. Each Tx $k$ randomly generates $M_k J_k$ codewords $\mathbf{x}_k(w_k, s_k) \in \mathcal{X}^n$ for each $(w_k, s_k) \in [1:M_k] \times [1:J_k]$ according to the distribution $P_k$ defined in \eqref{eqPk} where $\gamma_n$ is  determined later.  Upon observing message $W_k$ and secret key $S_k$,  Tx $k$ sends $\mathbf{x}_k(W_k,S_k)$. Each secret key $S_k$ is shared between only user pair $k$. Define a jointly typical set $\mathcal{A}_{\tau_k}^n$ as the following:
\begin{align} \label{eqn:jtset}
\mathcal{A}_{\tau_k}^n \coloneqq \left \{ (\mathbf{x}, \mathbf{y}) \in \mathcal{X}^n \times \mathcal{Y}_k^n: \log \frac{\bar{W}^{(k) \times n}(\mathbf{y}|\mathbf{x})}{\bar{W}^{(k) \times n}(\mathbf{y}|\mathbf{0})} > \tau_k \right \},
\end{align}
where
\begin{align}
\bar{W}^{(k)}(y|x_k) \coloneqq \sum_{x_{\mathcal{K}\backslash k}} W_{Y_k | X_{\mathcal{K}}}(y|x_{\mathcal{K}})\left( \prod_{j \in \mathcal{K}\backslash k} P_{j}(x_j) \right),
\end{align}
$\bar{W}^{(k) \times n}$ is the $n$-product channel of $\bar{W}^{(k)}$, and $\tau_k$ is determined later. Roughly speaking, the channel $\bar{W}^{(k)}$ can be interpreted as a p-to-p channel between user pair $k$ while the interference signals from the other users are treated as noise. Each Rx $k$ upon observing $\mathbf{y}_k$ and $S_k$ decodes as follows: 
\begin{itemize}
\item If there exists a unique message $w_k \in [1:M_k]$ such that $(\mathbf{x}_k(w_k, S_k), \mathbf{y}_k) \in \mathcal{A}_{\tau_k}^n$, outputs an estimate $\hat{W}_k = w_k$.
\item Otherwise, declares a decoding error.   
\end{itemize}

\subsubsection{Codebook Size for Reliable Communication}  The following lemma gives an upper bound on the average probability of decoding error over the random codebook ensemble for a certain codebook size. 
\begin{lemma}\label{lem:achvrate}
Fix $\epsilon \in (0,1)$. When $n$ is sufficiently large and $\gamma_n$ goes to zero as $n$ tends to infinity, for
\begin{align} \label{eqn:achvrate}
\log M_k = (1 - \epsilon)\alpha_k n \gamma_n D(W_{k}^{(k)} \| W_{0}^{(k)}), \quad \forall k \in \mathcal{K},
\end{align}
the average probability of decoding error over the random codebook ensemble is upper bounded as
\begin{align}
\mathbb{E}[P_e^n] \leq e^{-c n \gamma_n}
\end{align}
for a constant $c > 0$. 
\end{lemma}
\begin{proof}
Define $P_e^n(k) \coloneqq \Pr ( \hat{W}_k \neq W_k )$ for all $k \in \mathcal{K}$. From the union bound, we have $\mathbb{E}[P_e^n] \leq \sum_{k \in \mathcal{K}}\mathbb{E}[P_{e,k}^n]$. We analyze each term $\mathbb{E}[P_{e,k}^n]$. 

Consider Tx-Rx pair $k$ that utilize the encoding and decoding schemes described in Section \ref{subsubsec:scheme}. Due to the symmetry of the codebook ensemble, we can assume $W_k=1$ and $S_k=1$. Then, two types of decoding error events are defined as $E_{k,1} \coloneqq \{(\mathbf{x}_k(1, 1), \mathbf{y}_k) \notin \mathcal{A}_{\tau_k}^n\}$ and $E_{k,2} \coloneqq \{ \exists i \neq 1 \mbox{ s.t. } (\mathbf{x}_k(i, 1), \mathbf{y}_k) \in \mathcal{A}_{\tau_k}^n \}$. Also, we have $\mathbb{E}[P_{e,k}^n] \leq \mathbb{E}[\Pr(E_{k,1})] + \mathbb{E}[\Pr(E_{k,2})]$ by the union bound. Then, we obtain
\begin{align}
&\mathbb{E}[\Pr(E_{k,1})] + \mathbb{E}[\Pr(E_{k,2})] \nonumber \\
& = \Pr \left( (\mathbf{X}_k(1, 1), \mathbf{Y}_k) \notin \mathcal{A}_{\tau_k}^n \right)  \nonumber \\
& \quad + \Pr \left(  \exists i \neq 1 \mbox{ s.t. } (\mathbf{X}_k(i, 1), \mathbf{Y}_k) \in \mathcal{A}_{\tau_k}^n \right) \\ 
& \leq P_{W_{Y_k|X_{\mathcal{K}}}^{\times n} \cdot \prod_{k \in \mathcal{K}}P_{k}^{\times n}} \left( \log \frac{\bar{W}^{(k)\times n}(\mathbf{Y}_k|\mathbf{X}_k)}{\bar{W}^{(k)\times n}(\mathbf{Y}_k|\mathbf{0})} \leq \tau_k \right)  \nonumber \\
& \quad + M_k \sum_{\mathbf{x}_{\mathcal{K}}} \sum_{\mathbf{y}_k} W_{\boldsymbol{\alpha},\gamma_n}^{(k) \times n}(\mathbf{y}_k) \left( \prod_{k \in \mathcal{K}}P_{k}^{\times n}(\mathbf{x}_k) \right) \nonumber \\
&\quad \quad \times \mathds{1}\{ (\mathbf{x}_k, \mathbf{y}_k) \in \mathcal{A}_{\tau_k}^n \} \\ \label{eqn:ICtoDMC}
&= P_{\bar{W}^{(k) \times n} \cdot P_{k}^{\times n}} \left( \log \frac{\bar{W}^{(k)\times n}(\mathbf{Y}_k|\mathbf{X}_k)}{\bar{W}^{(k)\times n}(\mathbf{Y}_k|\mathbf{0})} \leq \tau_k \right) \nonumber \\
& \quad + M_k \sum_{\mathbf{x}_k} \sum_{\mathbf{y}_k} W_{\boldsymbol{\alpha},\gamma_n}^{(k) \times n}(\mathbf{y}_k) P_{k}^{\times n}(\mathbf{x}_k) \mathds{1}\{ (\mathbf{x}_k, \mathbf{y}_k) \in \mathcal{A}_{\tau_k}^n \}.
\end{align}
In fact, the right-hand side in \eqref{eqn:ICtoDMC} is an upper bound on the average probability of decoding error for a DMC $(\mathcal{X}, \bar{W}^{(k)}, \mathcal{Y}_k)$ when the same encoding and joint typicality decoding based on \eqref{eqn:jtset} are utilized (see the case of DMC with a warden \cite[Appendix D]{Bloch:16}). Thus, by applying the result in DMC \cite[Lemma 3]{Bloch:16}, for appropriately chosen $\tau_k$, we identify that if 
\begin{align} \label{eqn:tempptop}
\log M_k = (1 - \epsilon')\alpha_k n \gamma_n D(\bar{W}^{(k)}(\cdot | 1) \| \bar{W}^{(k)}(\cdot | 0))
\end{align}
for an arbitrarily small $\epsilon' > 0$, then $\mathbb{E}[P_e^n] \leq e^{-c n \gamma_n}$ for a constant $c > 0$. From some manipulations, we can check 
\begin{align}\label{eqn:approx1}
| D(W_{k}^{(k)} \| W_{0}^{(k)}) - D(\bar{W}^{(k)}(\cdot | 1) \| \bar{W}^{(k)}(\cdot | 0)) | = o(\gamma_n).
\end{align}
Thus, we have 
\begin{align}\label{eqn:approx2}
(1 - \epsilon')D(\bar{W}^{(k)}(\cdot | 1) \| \bar{W}^{(k)}(\cdot | 0)) \geq (1 - \epsilon)D(W_{k}^{(k)} \| W_{0}^{(k)}),
\end{align}
for an arbitrarily small constant $\epsilon > 0$ if $\gamma_n \rightarrow 0$ as $n \rightarrow \infty$. This completes the proof.
\end{proof} 
We note that it can be shown that $W_{0}^{(k)} = \bar{W}^{(k)}(\cdot | 0) + O(\gamma_n)$ and $W_{k}^{(k)} = \bar{W}^{(k)}(\cdot | 1) + O(\gamma_n)$. This shows that the effect of the sparse (i.e., $\gamma_n \rightarrow 0$ as $n \rightarrow \infty$) interfereing signals on each marginal p-to-p channel is negligible.

\subsubsection{Covert Communication in Channel Resolvability Perspective}
From the channel resolvability \cite{Han-Verdu, spectrum, Bloch:16, MAC}, one can make the channel output distribution at the warden to be close to i.i.d. process  (i.e., $\lim_{n \rightarrow \infty}D(\hat{Q}^n \| Q_{\boldsymbol{\alpha},\gamma_n}^{\times n}) = 0$) if a Tx sends sufficiently many codewords. The following lemma gives a sufficient size of codebooks at each Tx to guarantee this. The proof of this lemma is same with that of \cite[Lemma 3]{MAC}  and thus is omitted in this paper. 
\begin{lemma}\label{lem:achvkeyrate}
Fix $\epsilon \in (0,1)$. When $n$ is sufficiently large, for
\begin{align} \label{eqn:achvkeyrate}
\log M_kJ_k = (1 + \epsilon)\alpha_k n \gamma_n D(Q_{k} \| Q_{0}), \quad \forall k \in \mathcal{K},
\end{align}
the relative entropy between $\hat{Q}^n$ and $Q_{\boldsymbol{\alpha},\gamma_n}^{\times n}$ averaged over the random codebook ensemble is upper bounded as
\begin{align}
\mathbb{E} \left[ D(\hat{Q}^n \| Q_{\boldsymbol{\alpha},\gamma_n}^{\times n}) \right] \leq e^{-c n \gamma_n},  
\end{align}
for a constant $c > 0$. 
\end{lemma}

Now, we show that there exists a  coding scheme satisfying \eqref{eqn:error0}, \eqref{eqn:covert0}, \eqref{eqn:achvrate}, and \eqref{eqn:achvkeyrate}. First, assume that \eqref{eqn:achvrate} and \eqref{eqn:achvkeyrate} are satisfied. Then, by applying Markov's inequality, we obtain 
\begin{align} \label{eqn:Markov}
\Pr & ( \left( P_e^n < 4\mathbb{E}[P_e^n] \right) \cap \nonumber \\
& ( D(\hat{Q}^n \| Q_{\boldsymbol{\alpha},\gamma_n}^{\times n}) < 4 \mathbb{E}[D(\hat{Q}^n \| Q_{\boldsymbol{\alpha},\gamma_n}^{\times n})] ) ) \geq \frac{1}{2}.
\end{align} 
Thus, we can conclude that there exists a specific coding scheme that for sufficiently large $n$,
\begin{align}
P_e^n \leq e^{-c_1 n \gamma_n}, \\ \label{eqn:resolvability}
D(\hat{Q}^n \| Q_{\boldsymbol{\alpha},\gamma_n}^{\times n}) \leq e^{-c_2 n \gamma_n},
\end{align}
for constants $c_1 > 0$ and $c_2 > 0$.
Now we use the following lemma to show that $D(\hat{Q}^n \| Q_{0}^{\times n})$ can be arbitrarily small, whose proof can be checked in \cite[Appendix F]{MAC}.
\begin{lemma} \label{lem:close}
Assume that \eqref{eqn:resolvability} holds. Then, for sufficiently large $n$ and a constant $c_3 > 0$, we have
\begin{align} \label{eqn:resolvability2}
\left | D(\hat{Q}^n \| Q_{0}^{\times n}) - D(Q_{\boldsymbol{\alpha},\gamma_n}^{\times n} \| Q_{0}^{\times n})  \right | \leq e^{-c_3 n \gamma_n}.
\end{align}
\end{lemma}

By combining Lemma \ref{lem:basic} and Lemma \ref{lem:close}, we obtain
\begin{align} \label{eqn:covertnessbound}
\frac{n\gamma_n^2}{2}(1 - \sqrt{\gamma_n})\chi^2(\boldsymbol{\alpha}) - e^{-c_3 n \gamma_n} \leq D(\hat{Q}^n \| Q_{0}^{\times n}) \nonumber \\ 
\leq e^{-c_3 n \gamma_n} + \frac{n\gamma_n^2}{2}(1 + \sqrt{\gamma_n})\chi^2(\boldsymbol{\alpha}).
\end{align}
Thus, by choosing an appropriate sequence $\{\gamma_n\}_{n \geq 1}$, we can satisfy $\lim_{n \rightarrow \infty}D(\hat{Q}^n \| Q_{0}^{\times n}) = 0$, and conclude that there exists a coding scheme satisfying \eqref{eqn:error0}, \eqref{eqn:covert0}, \eqref{eqn:achvrate}, and \eqref{eqn:achvkeyrate}. Finally, by combining \eqref{eqn:achvrate}, \eqref{eqn:achvkeyrate}, and \eqref{eqn:covertnessbound}, for $k \in \mathcal{K}$, we have 
\begin{align}
\lim_{n \rightarrow \infty} \frac{\log M_k}{\sqrt{n D(\hat{Q}^n \| Q_{0}^{\times n})}} =  \frac{\alpha_k D(W_{k}^{(k)} \| W_{0}^{(k)})}{\sqrt{\chi^2(\boldsymbol{\alpha})/2}}
\end{align}
and
\begin{align}
\lim_{n \rightarrow \infty} \frac{\log M_k J_k}{\sqrt{n D(\hat{Q}^n \| Q_{0}^{\times n})}} = \frac{\alpha_k D(Q_{k} \| Q_{0})}{\sqrt{\chi^2(\boldsymbol{\alpha})/2}}.
\end{align}
In addition, these yield
\begin{align}
&\lim_{n \rightarrow \infty} \frac{\log J_k}{\sqrt{n D(\hat{Q}^n \| Q_{0}^{\times n})}} \nonumber \\ 
&= \frac{ \alpha_k \left[ D(Q_k\|Q_0) -  D(W_{k}^{(k)} \| W_{0}^{(k)}) \right]^{+}}{ \sqrt{\chi^2(\boldsymbol{\alpha})/2} }.
\end{align}

\subsection{Converse}\label{subsec:conv}
Consider a sequence of codes of block length $n$ for the DM-IC with a warden satisfying \eqref{eqn:error0} and \eqref{eqn:covert0}. Each message $W_k$ at Tx $k$ is encoded as a codeword $\mathbf{X}_k = (X_{k1}, X_{k2},\ldots,X_{kn}) \in \mathcal{X}^n$ for all $k \in \mathcal{K}$. Let the distribution of $X_{ki} \in \mathcal{X}$ as $P_{ki}$ that is defined as follows:
\begin{align}
P_{ki}(x) \coloneqq \frac{1}{M_kJ_k}\sum_{w_k=1}^{M_k}\sum_{s_k=1}^{J_k}\mathds{1}\{X_{ki}(w_k, s_k) = x\},
\end{align} 
for all $k \in \mathcal{K}$ and $i \in [1:n]$. We also denote $P_{ki}(1) = 1 - P_{ki}(0)$ as $p_{ki}^{(n)}$ that is defined as $p_{ki}^{(n)} \coloneqq \alpha_{ki}\gamma_n$ where $\alpha_{ki}, \gamma_n \geq 0$, and $\alpha_k \coloneqq \frac{1}{n}\sum_{i=1}^n \alpha_{ki}$. Without loss of generality, we assume $\sum_{k \in \mathcal{K}} \alpha_{k} = 1$. 

Let us denote $\hat{Q}_i$ as the marginal distribution of $\hat{Q}^n$ on the $i^{\mathrm{th}}$ component. Then, we have
\begin{align} \label{eqn:combination}
\hat{Q}_i(z) = \sum_{\mathcal{U} \subseteq \mathcal{K}} \left( \prod_{k \in \mathcal{U}}p_{ki}^{(n)} \right) \left( \prod_{k \in \mathcal{U}^c} \left( 1 - p_{ki}^{(n)} \right) \right) Q_{\mathcal{U}}(z).
\end{align}
Using equation \eqref{eqn:combination}, we now show $\lim_{n \rightarrow 0}p_{ki}^{(n)} = 0$ for all $k \in \mathcal{K}$ and $i \in [1:n]$, and derive a lower bound on $D(\hat{Q}^n \| Q_{0}^{\times n})$, which are used to prove the converse part of Theorem \ref{thm:region}. First, we obtain
\begin{align} \label{eqn:start}
D(\hat{Q}^n \| Q_{0}^{\times n}) &= - H(\mathbf{Z}) + \mathbb{E}_{\hat{Q}^n} \left[ \log \frac{1}{Q_{0}^{\times n}(\mathbf{Z})} \right] \\
&= \sum_{i=1}^n \left( - H(Z_i | Z^{i-1}) + \mathbb{E}_{\hat{Q}_i} \left[ \log \frac{1}{Q_{0}(Z_i)} \right] \right) \\ \label{eqn:Jensen}
&\geq \sum_{i=1}^n \left( - H(Z_i) + \mathbb{E}_{\hat{Q}_i} \left[ \log \frac{1}{Q_{0}(Z_i)} \right] \right) \\ \label{eqn:last}
&= \sum_{i=1}^n D(\hat{Q}_i \| Q_{0}).
\end{align}
Since we assume $\lim_{n \rightarrow \infty}D(\hat{Q}^n \| Q_{0}^{\times n}) = 0$, and the relative entropy is nonnegative, we have $\lim_{n \rightarrow \infty}D(\hat{Q}_i \| Q_{0}) = 0$ for all $i \in [1:n]$. Furthermore, by using Pinsker's inequality \cite{Cover:2006}, we obtain
\begin{align} 
\lim_{n\rightarrow \infty}| \hat{Q}_i(z) - Q_0(z)| = 0, \quad \forall z \in \mathcal{Z},
\end{align}
and thus
\begin{align} \label{eqn:Qlimit}
\lim_{n\rightarrow \infty}\hat{Q}_i(z) = Q_0(z), \quad \forall z \in \mathcal{Z}. 
\end{align}
In our setting, we assume that $Q_0$ cannot be represented as any convex combination of some $Q_{\mathcal{U}}$ for some $\mathcal{U} \subseteq \mathcal{K}$. From equations \eqref{eqn:combination} and \eqref{eqn:Qlimit}, we observe that this assumption is contradictory if there exists a sequence $p_{ki}^{(n)}$ such that $\lim_{n \rightarrow \infty} p_{ki}^{(n)} \neq 0$ for some $k$ and $i$. Hence, we can conclude that
\begin{align}
\lim_{n \rightarrow \infty} p_{ki}^{(n)} = 0, \quad \forall k \in \mathcal{K} \mbox{ and } \forall i \in [1:n].
\end{align}
A simple intuition behind is that if transmissions of symbol 1 are concentrated in some Txs at some specific time slots, detectability of the communication at the warden increases. This shows the importance of \textit{diffuse} signaling for the covert communication scenario \cite{Wang:16, PPM2}.

Now, we derive a lower bound on $D(\hat{Q}^n \| Q_{0}^{\times n})$. Define $\Delta_i^{(n)}(z) \coloneqq \hat{Q}_i(z) - Q_0(z)$. Then, from some manipulations \cite[Eq. (62)-(77)]{MAC}, we have
\begin{align} \label{eqn:relativelower}
D(\hat{Q}^n \| Q_{0}^{\times n}) \geq \sum_z \frac{1 - \phi^{(n)}(z)}{2 Q_0(z)} \sum_{i=1}^n \left( \Delta_i^{(n)}(z) \right)^2,
\end{align} 
where $\phi^{(n)}(z) = \max_{i \in [1:n]} \frac{\Delta_i^{(n)}(z)}{Q_0(z)}+ \frac{4|\Delta_i^{(n)}(z)|}{3Q_0(z)}$ and $\lim_{n \rightarrow \infty}\phi^{(n)}(z) = 0$.

\subsubsection{Outer Bound on the Covert Capacity Region} 

We derive the outer bound on the covert capacity region. By the standard technique, we have 
\begin{align}
&\log M_k \nonumber \\
&= H(W_k|W_{\mathcal{K} \backslash k}, S_{\mathcal{K} \backslash k}) \\
&= I(W_k ; \mathbf{Y}_k, S_k | W_{\mathcal{K} \backslash k}, S_{\mathcal{K} \backslash k}) + H(W_k | \mathbf{Y}_k, S_{\mathcal{K}}, W_{\mathcal{K} \backslash k})\\ 
&\leq I(W_k, S_k ; \mathbf{Y}_k | W_{\mathcal{K} \backslash k}, S_{\mathcal{K} \backslash k}) + H(W_k | \mathbf{Y}_k, S_{\mathcal{K}}, W_{\mathcal{K} \backslash k}) \\ \label{eqn:con_M_Fano}
&\leq I(W_k, S_k ; \mathbf{Y}_k | W_{\mathcal{K} \backslash k}, S_{\mathcal{K} \backslash k}) + H(\epsilon_{nk}) + \epsilon_{nk}\log M_k \\
&= \sum_{i=1}^{n} I(W_k, S_k ; Y_{ki}| Y_k^{i-1}, W_{\mathcal{K} \backslash k}, S_{\mathcal{K} \backslash k}) + H(\epsilon_{nk}) \nonumber \\
& \quad + \epsilon_{nk}\log M_k \\
&= \sum_{i=1}^{n} I(W_k, S_k, X_{ki} ; Y_{ki}| Y_k^{i-1}, W_{\mathcal{K} \backslash k}, S_{\mathcal{K} \backslash k}, X_{\{\mathcal{K} \backslash k\}i}) \nonumber \\
& \quad + H(\epsilon_{nk}) + \epsilon_{nk}\log M_k \\
&\leq \sum_{i=1}^{n} I(W_{\mathcal{K}}, S_{\mathcal{K}}, X_{ki}, Y_k^{i-1} ; Y_{ki} | X_{\{\mathcal{K} \backslash k\}i}) + H(\epsilon_{nk}) \nonumber \\
& \quad + \epsilon_{nk}\log M_k \\ \label{eqn:convlast}
&= \sum_{i=1}^{n} I(X_{ki} ; Y_{ki} | X_{\{\mathcal{K} \backslash k\}i}) + H(\epsilon_{nk}) + \epsilon_{nk}\log M_k,
\end{align}
where \eqref{eqn:con_M_Fano} is from Fano's inequality \cite{Cover:2006}, and $\epsilon_{nk} > 0$ is an arbitrarily small constant. Thus, we obtain
\begin{align} \label{eqn:convcodebook}
\log M_k \leq \frac{\sum_{i=1}^{n} I(X_{ki} ; Y_{ki}| X_{\{\mathcal{K} \backslash k\} i}) + H(\epsilon_{nk})}{1-\epsilon_{nk}}.
\end{align}
Since the operations of the Txs are independent at each time slot, we have
\begin{align}
&I( X_{ki} ; Y_{ki}| X_{\{\mathcal{K} \backslash k\} i} ) \nonumber \\ \label{eqn:1colstart}
&= \sum_{x_{\{\mathcal{K} \backslash k\} i}} \left( \prod_{j=1, \neq k}^{K} P_{ji}(x_{ji}) \right) \nonumber \\
& \quad \times I(X_{ki} ; Y_{ki}| X_{\{\mathcal{K} \backslash k\} i} = x_{\{\mathcal{K} \backslash k\} i})\\
&= \sum_{x_{\{\mathcal{K} \backslash k\} i}} \left( \prod_{j=1, \neq k}^{K} P_{ji}(x_{ji}) \right) I(X_{ki}, W_{Y_k| x_{\{\mathcal{K} \backslash k\} i}, X_k})\\ \label{eqn:ptopbound}
&\leq \sum_{x_{\{\mathcal{K} \backslash k\} i}} \left( \prod_{j=1, \neq k}^{K} P_{ji}(x_{ji}) \right) p_{ki}^{(n)} \nonumber \\
&\quad \times D\left (\left. W_{Y_k| x_{\{\mathcal{K} \backslash k\} i}, X_k}(\cdot | 1)  \right \| W_{Y_k| x_{\{\mathcal{K} \backslash k\} i}, X_k}(\cdot | 0) \right ) \\
&= \left( \prod_{j=1,\neq k}^{K} (1 - p_{ji}^{(n)}) \right) p_{ki}^{(n)} D(W_{k}^{(k)} \| W_{0}^{(k)}) \nonumber \\
& \quad + \sum_{x_{\{\mathcal{K} \backslash k\} i} \neq \mathbf{0}} \left( \prod_{j=1,\neq k}^{K} P_{ji}(x_{ji}) \right) p_{ki}^{(n)} \nonumber \\
&\quad \times D\left (\left. W_{Y_k| x_{\{\mathcal{K} \backslash k\} i}, X_k}(\cdot | 1)  \right \| W_{Y_k| x_{\{\mathcal{K} \backslash k\} i}, X_k}(\cdot | 0) \right ) \\
&\leq p_{ki}^{(n)} D\left (\left. W_{k}^{(k)} \right \| W_{0}^{(k)}\right ) + \sum_{x_{\{\mathcal{K} \backslash k\} i} \neq \mathbf{0}} \left( \prod_{j=1,\neq k}^{K} P_{ji}(x_{ji}) \right) \nonumber \\
&\quad \times p_{ki}^{(n)}  D\left (\left. W_{Y_k| x_{\{\mathcal{K} \backslash k\} i}, X_k}(\cdot | 1)  \right \| W_{Y_k| x_{\{\mathcal{K} \backslash k\} i}, X_k}(\cdot | 0) \right ) \\ \label{eqn:convmutu}
&= p_{ki}^{(n)} D(W_{k}^{(k)} \| W_{0}^{(k)}) + O(p_{\text{max}}^{(n)} \cdot p_{ki}^{(n)}),
\end{align}
where \eqref{eqn:ptopbound} is from the result in DMC with a warden \cite[Eq. (98)]{Bloch:16}, and $p_{\text{max}}^{(n)} \coloneqq \max_{\{k, i\}} p_{ki}^{(n)}$. Then, inequality \eqref{eqn:convcodebook} is now given as
\begin{align} \label{eqn:convcodebook2}
\log M_k \leq \frac{ \alpha_k n \gamma_n D(W_{k}^{(k)} \| W_{0}^{(k)}) + o(\alpha_k n \gamma_n) + H(\epsilon_{nk})}{1-\epsilon_{nk}}.
\end{align}
By combining \eqref{eqn:relativelower} and \eqref{eqn:convcodebook2}, we have \eqref{eqn:1col2start}-\eqref{eqn:Cauchy},
\begin{figure*}
\begin{align}\label{eqn:1col2start}
\frac{\log M_k}{\sqrt{nD(\hat{Q}^n \| Q_{0}^{\times n})}} &\leq \frac{ \sum_{i=1}^n \left( p_{ki}^{(n)} D(W_{k}^{(k)} \| W_{0}^{(k)}) + O(p_{\text{max}}^{(n)} \sum_{k \in \mathcal{K}} p_{ki}^{(n)}) \right) +H(\epsilon_{nk})}{(1 - \epsilon_{nk})\sqrt{n\sum_{z\in \mathcal{Z}} \frac{1 - \phi^{(n)}(z)}{2 Q_0(z)} \sum_{i=1}^n \left( \Delta_i^{(n)}(z) \right)^2}} \\
&= \frac{ \left(\sum_{k\in \mathcal{K}} \sum_{i=1}^n p_{ki}^{(n)}\right)  \left( D(W_{k}^{(k)} \| W_{0}^{(k)}) \frac{\sum_i p_{ki}^{(n)}}{\sum_k \sum_i p_{ki}^{(n)}}  + O(p_{\text{max}}^{(n)}) \right) + \frac{H(\epsilon_{nk})}{\sum_k \sum_i p_{ki}^{(n)}}}{(1 - \epsilon_{nk})\sqrt{n\sum_{z \in \mathcal{Z}} \frac{1 - \phi^{(n)}(z)}{2 Q_0(z)} \sum_{i=1}^n \left( \Delta_i^{(n)}(z) \right)^2}} \\
&= \frac{ D(W_{k}^{(k)} \| W_{0}^{(k)}) \frac{\sum_i p_{ki}^{(n)}}{\sum_k \sum_i p_{ki}^{(n)}}  + O(p_{\text{max}}^{(n)}) + \frac{H(\epsilon_{nk})}{\sum_k \sum_i p_{ki}^{(n)}}}
{(1 - \epsilon_{nk})\sqrt{\sum_{z \in \mathcal{Z}} \frac{1 - \phi^{(n)}(z)}{2 Q_0(z)} \frac{n\sum_i \left( \Delta_i^{(n)}(z) \right)^2}{\left(\sum_k \sum_i p_{ki}^{(n)}\right)^2}}} \\ \label{eqn:Cauchy}
&\leq \frac{ D(W_{k}^{(k)} \| W_{0}^{(k)}) \frac{\sum_i p_{ki}^{(n)}}{\sum_k \sum_i p_{ki}^{(n)}}  + O(p_{\text{max}}^{(n)}) + \frac{H(\epsilon_{nk})}{\sum_k \sum_i p_{ki}^{(n)}}}
{(1 - \epsilon_{nk})\sqrt{\sum_{z \in \mathcal{Z}} \frac{1 - \phi^{(n)}(z)}{2 Q_0(z)} \left( \frac{\sum_i \Delta_i^{(n)}(z)}{ \sum_k \sum_i p_{ki}^{(n)} }\right)^2 }},
\end{align}
\hrule
\end{figure*}
where \eqref{eqn:Cauchy} is due to Cauchy-Schwartz inequality. Then, by following the same steps with \cite[Eq. (107)-(114)]{MAC}, we finally obtain
\begin{align} \label{eqn:convrate}
\liminf_{n \rightarrow \infty}\frac{\log M_k}{\sqrt{nD(\hat{Q}^n \| Q_{0}^{\times n})}} \leq \frac{\alpha_k D(W_{k}^{(k)} \| W_{0}^{(k)})}{\sqrt{\chi^2(\boldsymbol{\alpha})/2}}, \quad \forall k \in \mathcal{K}.
\end{align}

\subsubsection{Converse Result for the Key Rate} 
To achieve the right-hand side in \eqref{eqn:convrate}, we must have
\begin{align}
\limsup_{n \rightarrow \infty}\frac{\log M_kJ_k}{\sqrt{nD(\hat{Q}^n \| Q_{0}^{\times n})}} \geq \frac{\alpha_k D(Q_{k} \| Q_{0})}{\sqrt{\chi^2(\boldsymbol{\alpha})/2}}, \quad \forall k \in \mathcal{K}.
\end{align}

The proof of the above inequality is same with that of \cite[Eq. (41)]{MAC} since the channel structure from the Txs to the warden is same with that of the DM-MAC with a warden \cite{MAC}. Thus, we omit the proof in this paper and refer the readers to check \cite[Eq. (116)-(155)]{MAC} for the concrete proof.

\section{Extensions}\label{sec:exten}
In this section, we extend the result for the BI DM-IC with a warden to several scenarios. The covert capacity regions of the non-binary input DM-IC and the Gaussian IC with a warden are presented in Sections \ref{subsec:nonbinary} and \ref{subsec:Gaussian}, respectively. For these two cases, we assume that a sufficiently long secret key is shared between every Tx-Rx pair, and we do not focus on the secret key length. Furthermore, we consider the BI DM-IC where the communication is required to be covert against $J$ wardens in Section \ref{subsec:Kwarden}. 

\subsection{Non-binary Input DM-ICs}\label{subsec:nonbinary}
The input alphabet at each Tx $k$ of the DM-IC is given as $\mathcal{X}_k = \{ 0, 1, \ldots, m_k \} $ where $0 \in \mathcal{X}_k$ is the ``off" input symbol. Define $m \coloneqq \max_{k \in \mathcal{K}}m_k$. Then, we introduce a $K$ by $m$ matrix $\boldsymbol{B} \coloneqq \{\beta_{ki}\}_{k \in \mathcal{K}, i \in [1:m]} \in [0,1]^{K \times m}$ such that $\sum_{i=1}^{m}\beta_{ki} = \sum_{i=1}^{m_k}\beta_{ki} = 1$ for all $k \in \mathcal{K}$. Roughly, $\alpha_k \beta_{ki} \gamma_n$ can be interpreted as the probability of sending symbol $i \neq 0$ in Tx $k$.
Theorem \ref{thm:region} can be generalized to the non-binary input case by  extending the proof steps in Section \ref{sec:proof} to involve the term $\boldsymbol{B}$. Thus, we omit the full proof of the extension for brevity, and present some changes in notations and main results.  

Similar to the case of the BI DM-IC, we define $Q_{k,i}(z) \coloneqq V_{  Z|X_k = i, X_{\mathcal{K} \backslash k} = \boldsymbol{0}  }(z)$ and $W_{k,i}^{(k)}(y) \coloneqq W_{ Y_k|X_k = i, X_{\mathcal{K} \backslash k} = \boldsymbol{0}  }(y) $. In addition, we define a chi-squared distance $\chi^2(\boldsymbol{\alpha}, \boldsymbol{B})$ as
\begin{align}
\chi^2(\boldsymbol{\alpha}, \boldsymbol{B}) \coloneqq \sum_z \frac{\left( \sum_{k \in \mathcal{K}}\sum_{i=1}^{m_k} \alpha_k \beta_{ki} Q_{k,i}(z) - Q_0(z) \right)^2}{Q_0(z)}.
\end{align}
Furthermore, similar to the distributions \eqref{eqPk} and \eqref{eqn:wardendist}, for $\gamma_n \in [0,1]$ and $k \in \mathcal{K}$, we define 
\begin{align}
P'_{k}(x) \coloneqq 
\begin{cases}
1- \alpha_k \beta_{ki} \gamma_n & x = 0 \\
 \alpha_k \beta_{ki} \gamma_n & x = i, \neq 0,
\end{cases} 
\end{align}
and
\begin{align}
Q_{\boldsymbol{\alpha}, \boldsymbol{B}, \gamma_n}(z) \coloneqq \sum_{x_{\mathcal{K}}} V_{Z | X_{\mathcal{K}}}(z | x_{\mathcal{K}}) \left( \prod_{k \in \mathcal{K}}P'_k(x_k) \right).
\end{align} 
Then, Lemma \ref{lem:basic} is extended into the following: 
\begin{align} 
\frac{\gamma_n^2}{2}(1 - \sqrt{\gamma_n})\chi^2(\boldsymbol{\alpha}, \boldsymbol{B}) \leq D(Q_{\boldsymbol{\alpha}, \boldsymbol{B}, \gamma_n} \| Q_0) \nonumber \\ \label{eqn:exkeylemma}
\leq \frac{\gamma_n^2}{2}(1 + \sqrt{\gamma_n})\chi^2(\boldsymbol{\alpha},\boldsymbol{B}).
\end{align}
Roughly, equation \eqref{eqn:exkeylemma} presents a bound on the number of each symbol at each Tx that can be reliably transmitted while satisfying the covertness constraint. By following the similar steps to those in Section \ref{sec:proof}, we obtain Theorem \ref{thm:nonbinary}. 
\begin{theorem}\label{thm:nonbinary}
For the $K$-user non-binary input DM-IC with a warden, the covert capacity region is the set of the rate tuple $R_{\mathcal{K}}^{\mathrm{NB}}$ satisfying 
\begin{align} \label{eqn:region}
R_k^{\mathrm{NB}} \leq \frac{\alpha_k \sum_{i=1}^{m_k} \beta_{ki} D(W_{k,i}^{(k)} \| W_{0}^{(k)})}{\sqrt{\chi^2(\boldsymbol{\alpha}, \boldsymbol{B})/2}}, \quad \forall k \in \mathcal{K}
\end{align}
for some $\boldsymbol{\alpha} \in [0,1]^K$ such that $\sum_{k \in \mathcal{K}}\alpha_k = 1$ and some $\boldsymbol{B} \coloneqq \{\beta_{ki}\}_{k \in \mathcal{K}, i \in [1:m]} \in [0,1]^{K \times m}$ such that $\sum_{i=1}^{m}\beta_{ki} = \sum_{i=1}^{m_k}\beta_{ki} = 1$ for all $k \in \mathcal{K}$.
\end{theorem}

\subsection{Gaussian ICs}\label{subsec:Gaussian}
Consider a $K$-user Gaussian IC with a warden. Let the channel gain from Tx $k$ to Rx $j$ be $g_{jk} \in \mathbb{R}$, and from Tx $k$ to the warden be $g_{wk} \in \mathbb{R}$. Then, at transmission time $i \in [1:n]$, the channel outputs at Rx $j$, $Y_{ji}$ and at the warden, $Z_{i}$ are given as 
\begin{align}
Y_{ji} = \sum_{k \in \mathcal{K}}g_{jk}X_{ki} + N_{ji},\\
Z_{i} = \sum_{k \in \mathcal{K}}g_{wk}X_{ki} + N_{wi},
\end{align}
where $N_{ji}$ and $N_{wi}$ are white Gaussian noise $\mathcal{N}(0,\sigma^2)$. We assume average transmit power constraints\footnote{ The covertness constraint restricts the transmit power to very low level. Thus, if $P_{\mathrm{AWGN}}$ is non-vanishing, the covert capacity region does not depend on the value of $P_{\mathrm{AWGN}}$.}
\begin{align}
\frac{1}{n}\sum_{i=1}^n x_{ki}^2(w_k) \leq P_{\mathrm{AWGN}}, \quad w_k \in [1:M_k],\: \forall k \in \mathcal{K}.
\end{align}
The covert capacity region of the Gaussian IC with a warden is shown in the following theorem.
\begin{theorem}\label{thm:Gaussiancapacity}
For the $K$-user Gaussian IC with a warden, the covert capacity region is the set of the rate tuple $R_{\mathcal{K}}^{\mathrm{AWGN}}$ satisfying
\begin{align}
R_{k}^{\mathrm{AWGN}} \leq \frac{1}{\lambda(\boldsymbol{\alpha})}\alpha_k g_{kk}^2, \quad \forall k \in \mathcal{K}
\end{align}
for some $\boldsymbol{\alpha} \in [0,1]^K$ such that $\sum_{k \in \mathcal{K}}\alpha_k = 1$, where $\lambda(\boldsymbol{\alpha})$ is defined as 
\begin{align}
\lambda(\boldsymbol{\alpha}) \coloneqq \sum_{k \in \mathcal{K}}\alpha_k g_{wk}^2. 
\end{align}
\end{theorem}
For the achievability, the codebook is generated randomly with a low symbol power (approximately order of $1/\sqrt{n}$) compared to the noise level at the warden, and each Rx $k$ utilizes the TIN scheme. Similar to the case of DM-IC, $\lambda(\boldsymbol{\alpha})$ varies with $\boldsymbol{\alpha}$ in general. However, if the channel to the warden is symmetric in the sense that $g_{w1}= \cdots = g_{wK}$, $\lambda(\boldsymbol{\alpha})$ is fixed and time-division approach is optimal. 
  
We first present a necessary condition on the covertness constraint, which is mainly used to prove the converse part of Theorem \ref{thm:Gaussiancapacity}. Then, we provide the achievability and the converse proofs for Theorem \ref{thm:Gaussiancapacity}.

\subsubsection{A Necessary Condition on the Covertness Constraint}\label{subsubsec:necessary}
By following the same steps with \eqref{eqn:start}-\eqref{eqn:last}, we obtain
\begin{align} \label{eqn:relativesingle}
D(\hat{Q}^n \| Q_{0}^{\times n}) \geq \sum_{i=1}^n D(\hat{Q}_i \| Q_{0}).
\end{align}
In addition, due to the convexity of relative entropy, we have 
\begin{align}
\sum_{i=1}^n D(\hat{Q}_i \| Q_{0}) \geq nD(\bar{Q} \| Q_0),
\end{align}
where $\bar{Q}$ is the distribution averaged over $\hat{Q}_i$ for $i=1,2,\ldots,n$. For a sequence of codes for the covert communication, let us define $P_{ki} \coloneqq \mathbb{E}[X_{ki}^2]$, $\bar{P}_k \coloneqq \frac{1}{n}\sum_{i=1}^n P_{ki}$, and $\bar{P}_{\mathrm{r}} \coloneqq \sum_{k \in \mathcal{K}} g_{wk}^2\bar{P}_k$, i.e., the average received power at the warden. In addition, without loss of generality, let $\bar{P}_k = \alpha_k \bar{P}$, where $\{ \alpha_k \}_{k \in \mathcal{K}} \in [0,1]^K$ and $\sum_{k \in \mathcal{K}}\alpha_k = 1$. Then, by mimicking the steps in \cite[Eq. (74)]{Wang:16}, we have
\begin{align}
\frac{\bar{P}_{\mathrm{r}}}{2\sigma^2} -\frac{1}{2}\log \frac{\bar{P}_{\mathrm{r}} + \sigma^2}{\sigma^2} \leq D(\bar{Q} \| Q_0).
\end{align}
Since $\lim_{n \rightarrow \infty}D(\bar{Q} \| Q_0) = 0$ implies $\lim_{n \rightarrow \infty}\bar{P}_{\mathrm{r}} = 0$, the above inequality yields 
\begin{align}
\frac{\bar{P}_{\mathrm{r}}^2}{4\sigma^4} + o(\bar{P}_{\mathrm{r}}^2) \leq D(\bar{Q} \| Q_0).
\end{align}
By combining the above inequality and \eqref{eqn:relativesingle}, we have
\begin{align}
\bar{P}_{\mathrm{r}}  \leq 2\sigma^2 \sqrt{\frac{D(\hat{Q}^n \| Q_{0}^{\times n})}{n}},
\end{align}
and thus
\begin{align} \label{eqn:Gaussiancore}
\frac{n \bar{P} \lambda(\boldsymbol{\alpha})}{2\sigma^2} \leq \sqrt{nD(\hat{Q}^n \| Q_{0}^{\times n})}.  
\end{align}

\subsubsection{Achievability}
Fix $\boldsymbol{\alpha} \in [0,1]^K$ such that $\sum_{k \in \mathcal{K}}\alpha_k = 1$. In a similar way as in Section \ref{subsubsec:scheme}, each Tx $k$ uses random coding where each of $M_kJ_k$ ($J_k$ is assumed to be sufficiently large) codewords is generated according to the distribution $\mathcal{N}(0,\alpha_k P)$ where $P$ is determined later. Each Rx $k$ decodes its message while TIN. Since $\lim_{n \rightarrow \infty}P$ must be zero for covert communication from \eqref{eqn:covertness} and \eqref{eqn:Gaussiancore}, the interference power at each Rx is negligible compared to that of the background noise when $n$ is sufficiently large. By using the TIN scheme, we can treat this situation as $K$ parallel Gaussian channels \cite[Section 9.4]{Cover:2006} with a common covertness constraint where each noise variance at Rx $k$ is given as $\sigma^2 / g_{kk}^2 + o(n^{-1/2})$. Thus, by applying the result in AWGN channel with a warden \cite[Section V]{Wang:16}, we identify that 
\begin{align}
\frac{\log M_k}{\sqrt{nD(\hat{Q}^n \| Q_{0}^{\times n})}} &= (1 - \epsilon)\frac{n g_{kk}^2\alpha_k P}{2 \sigma^2 \sqrt{nD(\hat{Q}^n \| Q_{0}^{\times n})}}
\end{align}
is achievable for all $k \in \mathcal{K}$ and for an arbitrarily small $\epsilon > 0$ if $P \leq \frac{2\sigma^2}{\lambda(\boldsymbol{\alpha})}\sqrt{\frac{D(\hat{Q}^n \| Q_{0}^{\times n})}{n}}$. By choosing the maximum value of $P$, we complete the achievability proof.

\subsubsection{Converse}
Consider a sequence of codes for covert communication. For the converse proof, we start from \eqref{eqn:convlast}. Then, we obtain
\begin{align}
\sum_{i=1}^{n} I(X_{ki} ; Y_{ki} | X_{\{\mathcal{K} \backslash k \}i }) 
&= \sum_{i=1}^{n} I(X_{ki} ; g_{kk}X_{ki} + N_{ki}) \\
&= \sum_{i=1}^{n} h(g_{kk}X_{ki} + N_{ki}) - h(N_{ki}) \\ \label{eqn:differentialmax}
&\leq \sum_{i=1}^{n} \frac{1}{2} \log \frac{g_{kk}^2 P_{ki} + \sigma^2}{\sigma^2}\\
&\leq \sum_{i=1}^{n} \frac{g_{kk}^2P_{ki}}{2\sigma^2} \\ \label{eqn:Gaussianconvlast}
&= \frac{n \bar{P}g_{kk}^2\alpha_k}{2\sigma^2},
\end{align}
where \eqref{eqn:differentialmax} is because Gaussian distribution maximizes differential entropy under fixed power. By combining \eqref{eqn:Gaussiancore} and \eqref{eqn:Gaussianconvlast}, we have
\begin{align}
\frac{\log M_k}{\sqrt{nD(\hat{Q}^n \| Q_{0}^{\times n})}} \leq \frac{1}{\lambda(\boldsymbol{\alpha})}\alpha_k g_{kk}^2,
\end{align}
for all $k \in \mathcal{K}$, which completes the proof.

\subsection{$J$-Warden}\label{subsec:Kwarden}
Consider a BI DM-IC with $J \geq 2$ non-colluding wardens where each warden $j$ monitors the communication through a $K$-user DM-MAC $(\mathcal{X}_{\mathcal{K}}, V_{Z_j|X_{\mathcal{K}}}, \mathcal{Z}_j)$. The channel output distribution of $\mathbf{Z}_j$ at warden $j$ is written as $Q_{0}^{(j) \times n}$ when no communication occurs, and as $\hat{Q}^{(j)n}$ when communication takes place. We define $Q_{\mathcal{U}}^{(j)}(z) \coloneqq V_{Z_j|X_{\mathcal{K}}}(z|b(\mathcal{U}))$, and if $\mathcal{U} = \{i\}$ for $i\in \mathcal{K}$, we write $Q_{i}^{(j)}$. The absolute continuity described in Section \ref{sec:problem} is assumed with respect to all the wardens. We define the set $\mathcal{J} \coloneqq [1:J]:=\{1,\cdots, J\}$. The covertness constraint is given as 
\begin{align} 
\lim_{n \rightarrow \infty} D ( \hat{Q}^{(j)n} \| Q_{0}^{(j)\times n} ) = 0, \quad \forall j \in \mathcal{J}.
\end{align}     
For simplicity, for a specific coding scheme, we define $D_{\mathrm{max}}^n \coloneqq \max_{j \in \mathcal{J}} D ( \hat{Q}^{(j)n} \| Q_{0}^{(j)\times n})$.
Then, the covert capacity region of the DM-IC with $J$ wardens is formally defined as the following.
  \begin{definition}
We say that a tuple pair $ (R_{\mathcal{K}}^{\mathrm{J}}, L_{\mathcal{K}}^{\mathrm{J}}) \in \mathbb{R}_{+}^{2K}$ is achievable for the $K$-user DM-IC with $J$ wardens if there exists a sequence of codes satisfying the following:
\begin{align}
 \liminf_{n \rightarrow \infty} \frac{\log M_k}{\sqrt{nD_{\mathrm{max}}^n}}  &\geq R_k^{\mathrm{J}}, \quad \forall k \in \mathcal{K},
\end{align}
\begin{align}
\limsup_{n \rightarrow \infty} \frac{\log J_k}{\sqrt{nD_{\mathrm{max}}^n}} \leq L_k^{\mathrm{J}}, \quad \forall k \in \mathcal{K},
\end{align}
\begin{align} \label{eqn:kerror0}
\lim_{n \rightarrow \infty} P_e^n = 0, 
\end{align}
and
\begin{align} \label{eqn:kcovert0}
\lim_{n \rightarrow \infty} D_{\mathrm{max}}^n = 0.
\end{align}
The covert capacity region of the $K$-user DM-IC with $J$ wardens is defined as the closure of the set $ \{ R_{\mathcal{K}}^{\mathrm{J}} \in \mathbb{R}_{+}^{K} : (R_{\mathcal{K}}^{\mathrm{J}},L_{\mathcal{K}}^{\mathrm{J}}) \mbox{ is achievable for some } L_{\mathcal{K}}^{\mathrm{J}} \}$.
\end{definition}

For $\boldsymbol{\alpha} \coloneqq \{ \alpha_k \}_{k \in \mathcal{K}} \in [0,1]^K$ such that $\sum_{k \in \mathcal{K}}\alpha_k = 1$, we define the chi-squared distances between $\sum_{k \in \mathcal{K}} \alpha_k Q_{k}^{(j)}(z)$ and $Q_{0}^{(j)}(z)$ as
\begin{align}
\chi_{j}^2(\boldsymbol{\alpha}) \coloneqq \sum_z \frac{\left( \sum_{k \in \mathcal{K}} \alpha_k Q_{k}^{(j)}(z) - Q_{0}^{(j)}(z) \right)^2}{Q_{0}^{(j)}(z)}.
\end{align}
For brevity, we also define $\chi_{\mathrm{max}}^2(\boldsymbol{\alpha}) \coloneqq \max_{j \in \mathcal{J}} \chi_{j}^2(\boldsymbol{\alpha}).$
Furthermore, we define the channel output distributions at each warden $j$ induced by the input distribution $P_k$ at each Tx $k$ as the following:
\begin{align}  \label{eqn:Kwardendist}
Q_{\boldsymbol{\alpha},\gamma_n}^{(j)}(z) \coloneqq \sum_{x_{\mathcal{K}}} V_{Z_j | X_{\mathcal{K}}}(z_j | x_{\mathcal{K}}) \left( \prod_{k \in \mathcal{K}}P_k(x_k) \right).
\end{align} 
Then, according to Lemma \ref{lem:basic}, the following holds:
\begin{align} 
\frac{\gamma_n^2}{2}(1 - \sqrt{\gamma_n})\chi_j^2(\boldsymbol{\alpha}) \leq D(Q_{\boldsymbol{\alpha},\gamma_n}^{(j)} \| Q_{0}^{(j)}) \leq \frac{\gamma_n^2}{2}(1 + \sqrt{\gamma_n})\chi_j^2(\boldsymbol{\alpha}).
\end{align}
Futhermore, for $(X_{\mathcal{U}}, Z_j) \in \mathcal{X}^{|\mathcal{U}|} \times \mathcal{Z}_j$ for $\mathcal{U} \subseteq \mathcal{K}$ and $|\mathcal{U}| \neq 0$, with the joint distribution $V_{Z_j|X_{\mathcal{U}}}(\prod_{k \in \mathcal{U}}P_k)$, we have
\begin{align}
I(X_{\mathcal{U}} ; Z_j) = \sum_{k \in \mathcal{U}} \alpha_k \gamma_n D(Q_{k}^{(j)} \| Q_{0}^{(j)}) + O(\gamma_n^2).
\end{align}  

The covert capacity region of the $K$-user DM-IC with $J$ wardens is characterized in the following theorem. 
\begin{theorem}\label{thm:Kregion}
For the $K$-user DM-IC with $J$ wardens, the covert capacity region is the set of the rate tuple $R_{\mathcal{K}}^{\mathrm{J}}$ satisfying 
\begin{align} \label{eqn:Kregion}
R_k^{\mathrm{J}} \leq \frac{\alpha_k D(W_{k}^{(k)} \| W_{0}^{(k)})}{\sqrt{\chi_{\mathrm{max}}^2(\boldsymbol{\alpha})/2}}, \quad \forall k \in \mathcal{K}
\end{align}
for some $\boldsymbol{\alpha} \in [0,1]^K$ such that $\sum_{k \in \mathcal{K}}\alpha_k = 1$.
For $R_{\mathcal{K}}^{\mathrm{J}}$ satisfying  \eqref{eqn:Kregion} with equalities, a sufficient and necessary condition on the tuple $L_{\mathcal{K}}^{\mathrm{J}}$ for $(R_{\mathcal{K}}^{\mathrm{J}}, L_{\mathcal{K}}^{\mathrm{J}})$ to be achievable is 
\begin{align} \label{eqn:Kratekeyrate}
L_k^{\mathrm{J}} \geq \frac{\alpha_k [ \max_{j \in \mathcal{J}}D(Q_{k}^{(j)} \| Q_{0}^{(j)}) - D(W_{k}^{(k)} \| W_{0}^{(k)})]^{+}}{\sqrt{\chi_{\mathrm{max}}^2(\boldsymbol{\alpha})/2}} 
\end{align}
for all $k \in \mathcal{K}$. Thus, if $\max_{j \in \mathcal{J}}D(Q_{k}^{(j)} \| Q_{0}^{(j)}) \leq D(W_{k}^{(k)} \| W_{0}^{(k)})$ (i.e., roughly all the channels from Tx $k$ to the wardens are worse than the channel from Tx $k$ to Rx $k$), a secret key between user pair $k$ is unnecessary.  
\end{theorem}
The term $\max_{j \in \mathcal{J}}D(Q_{k}^{(j)} \| Q_{0}^{(j)})$ in \eqref{eqn:Kratekeyrate} is because each Tx has to control a codebook size with respect to the covertness constraint (or the channel resolvability) against all the wardens. 

\subsubsection{ Achievability }
The number of wardens does not affect the channel reliability, and thus Lemma \ref{lem:achvrate} still holds. In the channel resolvability perspective, each of the Txs has to control the number of codewords to make the induced channel output distributions at all the wardens as approximately i.i.d. From this fact, Lemma \ref{lem:achvkeyrate} is changed as follows.
\begin{lemma}\label{lem:kachvkeyrate}
Fix $\epsilon \in (0,1)$. When $n$ is sufficiently large, for
\begin{align} \label{eqn:kachvkeyrate}
\log M_kJ_k = (1 + \epsilon)\alpha_k n \gamma_n \max_{j \in \mathcal{J}}D(Q_{k}^{(j)} \| Q_{0}^{(j)}), \quad \forall k \in \mathcal{K},
\end{align}
the relative entropy between $\hat{Q}^{(j)n}$ and $Q_{\boldsymbol{\alpha},\gamma_n}^{(j)\times n}$ averaged over the random codebook ensemble is upper bounded as
\begin{align}
\max_{j \in \mathcal{J}}\mathbb{E} \left[ D(\hat{Q}^{(j)n} \| Q_{\boldsymbol{\alpha},\gamma_n}^{(j)\times n}) \right] \leq e^{-c n \gamma_n},  
\end{align}
for a constant $c > 0$. 
\end{lemma}

By following the similar steps to equations \eqref{eqn:Markov} to \eqref{eqn:resolvability2}, we can verify that there exists a specific code satisfying \eqref{eqn:kerror0}, \eqref{eqn:kcovert0}, \eqref{eqn:achvrate}, \eqref{eqn:kachvkeyrate}, and the following inequalities:   
\begin{align} \label{eqn:kcovertnessbound}
\frac{n\gamma_n^2}{2}(1 - \sqrt{\gamma_n})\chi_j^2(\boldsymbol{\alpha}) - e^{-c' n \gamma_n} \leq D(Q_{\boldsymbol{\alpha},\gamma_n}^{(j)\times n} \| Q_{0}^{(j)\times n}) \nonumber \\
\leq e^{-c' n \gamma_n} +  \frac{n\gamma_n^2}{2}(1 + \sqrt{\gamma_n})\chi_j^2(\boldsymbol{\alpha})
\end{align}
for a constant $c' > 0$. Combining \eqref{eqn:achvrate}, \eqref{eqn:kachvkeyrate}, and \eqref{eqn:kcovertnessbound}, for all $k \in \mathcal{K}$, we obtain 
\begin{align}
\lim_{n \rightarrow \infty} \frac{\log M_k}{\sqrt{n D_{\mathrm{max}}^n}} = \frac{\alpha_k D(W_{k}^{(k)} \| W_{0}^{(k)})}{\sqrt{\chi^2(\boldsymbol{\alpha})/2}}, 
\end{align}
and
\begin{align}
\lim_{n \rightarrow \infty} \frac{\log M_k J_k}{\sqrt{n D_{\mathrm{max}}^n}} = \frac{\alpha_k \max_{j \in \mathcal{J}}D(Q_{k}^{(j)} \| Q_{0}^{(j)})}{\sqrt{\chi_{\mathrm{max}}^2(\boldsymbol{\alpha})/2}}.
\end{align}
In addition, these yield
\begin{align}
&\lim_{n \rightarrow \infty} \frac{\log J_k}{\sqrt{n D_{\mathrm{max}}^n}} \nonumber \\ 
&= \frac{ \alpha_k \left[ \max_{j \in \mathcal{J}}D(Q_{k}^{(j)} \| Q_{0}^{(j)}) -  D(W_{k}^{(k)} \| W_{0}^{(k)}) \right]^{+}}{ \sqrt{\chi_{\mathrm{max}}^2(\boldsymbol{\alpha})/2} },
\end{align}
which ends the achievability proof.

\subsubsection{ Converse }
The converse proof for inequality \eqref{eqn:Kregion} is same with the case of a single warden. Consider the converse proof for inequality \eqref{eqn:Kratekeyrate}. In the presence of $J$ wardens, the lower bound on the number of codewords at each Tx for reliable and covert communication is duplicated to $J$ lower bounds corresponding to $J$ wardens.
. For the proof, we simply change \cite[Eq. (118)]{MAC} by 
\begin{align}
\log M_k J_k \geq I(\mathbf{X}_k ; \mathbf{Z}_j), \quad \forall j \in \mathcal{J}.
\end{align}  
For each $j \in \mathcal{J}$, we follow the same steps to the case of the DM-IC with a warden. Then, for each $k \in \mathcal{K}$, we have
\begin{align}
\limsup_{n \rightarrow \infty}\frac{\log M_kJ_k}{\sqrt{nD_{\mathrm{max}}^n}} \geq \frac{\alpha_k D(Q_{k}^{(j)} \| Q_{0}^{(j)})}{\sqrt{\chi_{\mathrm{max}}^2(\boldsymbol{\alpha})/2}}, \quad \forall j \in \mathcal{J}.
\end{align}
Then, by combining $J$ lower bounds and following the same step with that of the DM-IC with a warden, we end the converse proof for \eqref{eqn:Kratekeyrate}.

\section{Conclusion}\label{sec:concl}
In this paper, we characterized the covert capacity region for $K$-user DM-ICs and Gaussian ICs where the communication is monitored by possibly many wardens. We showed that a p-to-p based scheme with TIN is an optimal strategy. This is because the covertness constraint highly restricts the transmissions of non-zero symbols (in DM-ICs) or the transmit powers  (in Gaussian ICs), and thus the additional channel randomness due to the sparse transmissions of non-zero symbols (or the interfering signals with very low power) is negligible compared to the intrinsic channel randomness present when each Tx keeps silent. For DM-ICs, we showed that if the channel  between a  Tx-Rx pair is better than that between the Tx and the wardens in a a certain way, then the secret key is not necessary between the user pair. 

As mentioned in Remark \ref{rem:absolute}, in the case that $W_{\mathcal{U}}^{(k)} \nll W_{0}^{(k)}$ for some $k \in \mathcal{K}$ and for some $\mathcal{U} \subseteq \mathcal{K}$, the covert capacity region is not characterized. In this case, a more delicate scheme beyond TIN seems to be needed, which would be an interesting further work.

%

\bibliographystyle{IEEEtran}
\bibliography{References}

\end{document}